\documentclass[12pt]{article}
\usepackage{amsmath}
\usepackage{graphicx,psfrag,epsf}
\usepackage{enumerate}
\usepackage[title]{appendix}

\usepackage{bm} 
\usepackage{amsmath} 
\usepackage{comment} 
\usepackage{apacite} 
\newcommand{\bs}{\bigskip} 
\usepackage{mediabb} 
\usepackage{amsthm}

\def\diag{\mathop{\rm diag}}

\def\where{\mathop{\rm where}}
\def\s.t.{\mathop{\rm s.t.}}
\def\tr{\mathop{\rm tr}}
\def\max{\mathop{\rm max}}
\def\min{\mathop{\rm min}}
\def\bdiag{\mathop{{\rm b}{\text -}{\rm diag}}}

\newcommand{\blind}{0}

\newtheorem{prop}{Proposition}[section]

\begin{document}

\def\spacingset#1{\renewcommand{\baselinestretch}%
{#1}\small\normalsize} \spacingset{1}


\if0\blind
{
  \title{\bf  Visualizing class specific heterogeneous tendencies in categorical data}
  \author{Mariko Takagishi \hspace{.2cm}\\
    Graduate School of Culture and Information Science, Doshisha University\\
    and \\
    Michel van de Velden \\
    Department of Econometrics, Erasmus University Rotterdam}
  \maketitle
} \fi

\if1\blind
{
  \bigskip
  \bigskip
  \bigskip
  \begin{center}
    {\LARGE\bf Visualizing class-specific heterogeneous tendencies in categorical data}
  \end{center}
  \medskip
} \fi

\bigskip
\begin{abstract}
In multiple correspondence analysis, both individuals (observations) and categories can be represented in a biplot that jointly depicts the relationships across categories or individuals, as well as the associations between them. Additional information about the individuals can enhance interpretation capacities, such as by including categorical variables for which the interdependencies are not of immediate concern, but that facilitate the interpretation of the plot with respect to relationships between individuals and categories. This article proposes a new method for adding such information, according to a multiple-set cluster correspondence analysis approach that identifies clusters specific to classes, or subsets of the data that correspond to the categories of the additional variables. The proposed method can construct a biplot that depicts heterogeneous tendencies of individual members, as well as their relationships with the original categorical variables. A simulation study to investigate the performance of this proposed method and an application to data regarding road accidents in the United Kingdom confirms the viability of this approach.  
\end{abstract}

\noindent%
\vfill

\newpage
\spacingset{1.45} 

%
\section{Introduction}
\label{sec:intro}
%


Correspondence analysis (CA) and multiple correspondence analysis (MCA) are popular methods that support visual interpretations of the associations among categorical variables \cite<e.g.,>{greenacre1984theory}. In MCA, obtained quantifications of categories and individuals can be depicted in a biplot, which indicates not only the associations among categories and among individuals but also those between individuals and categories \cite<e.g.,>{greenacre1993biplots, gower1995biplots}. 

In an MCA biplot, if many individuals choose the same two categories, the quantifications for these categories and corresponding individuals tend to locate in close proximity. Therefore, an MCA biplot enables us to visually identify individuals with similar category choice tendencies. Moving beyond this benefit, adding pertinent external information about individuals can enhance interpretations of MCA biplots. By external information, we refer to information that might not be of use for the estimation of the coordinates, but that may be useful for interpreting the resulting biplot.

Several studies describe ways to incorporate external information about individuals into an MCA biplot \cite<e.g.,>{yanai1986some,yanai1988partial,bockenholt1990canonical,takane1991relationships,van1992equality,bockenholt1994linear,yanai2002partial,hwang2005simultaneous}. \citeA{hwang2002generalized} also show that various objectives for incorporating the external information can be generalized into a linear constraint framework. Here, we focus specifically on external information that consists of a set of categorical variables, and we refer to subsets of these data that correspond to the categories of the external information as classes.

To visually explicate how individuals' tendencies differ depending on these classes, we could integrate external variables before applying MCA, but this approach transforms the information, such that it is no longer external, the information is no longer external and instead becomes part of the original analysis. As an alternative approach, we might seek to establish individual quantifications (i.e., points) visually, according to the classes. For example, if gender is an external variable, points corresponding to men can be colored black, and those corresponding to women are red. This approach incorporates external information corresponding to only one categorical variable at the time. Another option would be to obtain average quantifications for each class. By plotting these average points, as well as the category points of the original (non-external) variables, we can depict the relationship between the external information and the categories. We refer to this notion as the averaging approach.

Yet the averaging approach only reveals the average tendencies of many individuals within a class, obscuring their heterogeneous tendencies. When a relatively small group in a class has a strong tendency toward a particular category that the majority group in the class does not select, this preference would not be visible in a biplot that relies on an averaging approach. Despite representing a minority, such tendencies could be interesting to consider, especially to characterize tendencies by class.

Therefore, we propose a new approach to find class-specific clusters and depict them together with the categories of the (original) variables. The result is a visual depiction of the categories (i.e., category quantifications), together with points that represent clusters for the different classes of data. With this visualization, we can identify different heterogeneous tendencies within a class in a single MCA biplot, as well as perceive the relationships among classes that correspond to the categories of external variables.

The remainder of this paper is organized as follows. In Section 2, we introduce our proposed method and its relationship with existing approaches, including the linear row constraint framework. Then in Section 3, we compare a biplot obtained using the averaging approach and one obtained using our proposed method. The simulation study in Section 4 appraises the proposed method in various external information scenarios; the application of our method to empirical data in Section 5 confirms its appeal.

%
\section{Multiple Set Cluster CA}
\label{sec:prop}
%
In this section, we introduce our approach, which we call multiple-set cluster CA (MSCCA), as an extension of several existing methods, such as cluster CA \cite{van2017cluster}, CA, and the linear row constraint framework.

\subsection{The MSCCA objective function}
\label{sec:prop-obj}

Suppose that we have $N$ observations of $m$ categorical variables, and in conjunction, that, for the same $N$ observations, we have $H$ additional categorical variables  that contain external information. We refer to these $H$ additional variables as supplementary variables. To formulate the MSCCA objective function, we specify some notation upfront. 

Let $q_j$ $(j=1,\ldots,m)$ be the number of categories for the $j$th variable, and let $Q=\sum_{j=1}^mq_j$. We create dummy matrices $\bm{Z}_j$ for the $m$ categorical variables using the categorical data, so the rows of $\bm{Z}_j$ are $(q_j\times 1)$ vectors $\bm{z}_{ji}=(z_{ji\ell})$ $(i=1,\ldots,n\,;\,\ell=1,\ldots,q_j)$, where $z_{ji\ell}=1$ if individual $i$ chooses the $\ell$th category in the $j$th variable, and the other elements are $0$. Similarly, we create dummy matrices for the $H$ supplementary variables, with $\bm{V}_h=(v_{his})$ $(h=1,\ldots,H;\,s=1,\ldots,r_h)$, where $r_h$ is the number of categories for the $h$th supplementary variable. 

In addition, let $K_{hs}$ be the number of clusters for the $s$th category (class) of the $h$th supplementary variable, with $K_h=\sum_{s=1}^{r_h}K_{hs}$. Let $\bm{B}_j$ be the $q_j\times p$ quantification matrix for the categories of the $j$th variable, where $p$ denotes the number of dimensions, and let $\bm{U}_{h}$ and $\bm{G}_{h}$ be $N\times K_{h}$ cluster indicator matrices and $K_{h}\times p$ quantification matrices for cluster centers in the $h$th supplementary variable, respectively. The objective function of MSCCA then can be defined as
\begin{gather}
\min_{\bm{U}_{h}, \bm{G}_{h}, \bm{B}_j} \phi(\bm{U}_{h}, \bm{G}_{h}, \bm{B}_j\,|\,\bm{Z}_{j},\bm{V}_h)=\frac{1}{NHm}\sum_{j=1}^m\sum_{h=1}^H\|\bm{U}_{h}\bm{G}_{h}-\bm{Z}_{j}\bm{B}_j\|^2\label{eq:obj}\\
{\rm s.t.}\quad\frac{1}{Nm}\sum_{j=1}^m\bm{B}_j^{\prime}\bm{Z}_{j}^{\prime}\bm{Z}_{j}\bm{B}_j=\bm{I}_p,\quad\bm{J}_N\bm{U}_{h}\bm{G}_{h}=\bm{U}_{h}\bm{G}_{h}\nonumber\\
\where\,\,\,\underset{(N\times K_h)}{\bm{U}_h}=\left(\begin{array}{ccc}
\bm{u}_{h11}^{\prime} & \cdots & \bm{u}_{h1r_h}^{\prime}\\
\vdots & \ddots & \vdots\\
\bm{u}_{hN1}^{\prime} & \cdots & \bm{u}_{hNr_h}^{\prime}
\end{array}\right)\nonumber\\
\bm{u}_{his}=(u_{his1},\ldots,u_{hisK_{hs}})^{\prime}\nonumber\\
\s.t.\quad\begin{cases}
u_{hisk}\in\{0,1\},\,\,(k=1,\ldots,K_{hs}),\,\,\,\sum_{k=1}^{K_{hs}}u_{hisk}=1 &(v_{his}=1)\\
u_{hisk}=0,\,\,\,(k=1,\ldots,K_{hs}) &(v_{his}=0)
\end{cases}\label{eq:constU}\\
(i=1,\ldots,n\,;\,s=1,\ldots,r_h\,;\,h=1,\ldots,H)\nonumber.
\end{gather}
Here, $\bm{J}_{N}=\bm{I}_{N}-N^{-1}\bm{1}_{N}\bm{1}_{N}^{\prime}$ is the centering matrix, $\bm{I}_{N}$ is an $N\times N$ identity matrix, and $\bm{1}_{N}$ is an $N\times 1$ vector of ones. When we estimate parameters, the number of clusters $K_{hs}$ ($h=1,\ldots,H\,;\,s=1,\ldots,r_h$) must be pre-specified.

The constraint on $\bm{U}_h$ in Equation $(\ref{eq:constU})$ defines a two-level hierarchical cluster structure. Specifically, for each supplementary variable $h$, individuals first are divided into $r_h$ \emph{known} classes, corresponding to the categories of the variable as indicated by $\bm{u}_{hi1},\ldots,\bm{u}_{hir_h}$. Then within each class $s$ ($s=1,\ldots,r_h$), individuals are assigned to $K_{hs}$ unknown clusters as indicated by $u_{his1},\ldots,u_{ihsK_{hs}}$.

We can illustrate the construction of $\bm{U}_h$ with a small example. Suppose that we have five observations and that one supplementary variable, (e.g., $h$), corresponds to gender. In addition, assume we want to find two clusters for the males and one cluster for females, so that $K_{h1} = 2$ and $K_{h2} = 1$. Let observations $i = 1,3,5$ be males where $i = 1,3$ are in the first male cluster and $i = 5$ is in the second one, individuals $i = 2, 4$ are females.

When we consider the cluster indicator vector for individual $i = 1$, $\bm{u}_{h1}$, because we partition the data by gender, the vector is split as $\bm{u}_{h1}^{\prime}=(\bm{u}_{h11}^{\prime}, \bm{u}_{h12}^{\prime})$, where $\bm{u}_{h11}$ and $\bm{u}_{h12}$ denote the cluster indicator vectors of individual $i$ in the male and female classes, respectively. Performing this partitioning for all individuals, we obtain
\begin{align*}
	\bm{U}_{h}=\left(\begin{array}{cc}
		\bm{u}_{h11}^{\prime} & \bm{u}_{h12}^{\prime}\\
		\bm{u}_{h21}^{\prime} & \bm{u}_{h22}^{\prime}\\
		\bm{u}_{h31}^{\prime} & \bm{u}_{h32}^{\prime}\\
		\bm{u}_{h41}^{\prime} & \bm{u}_{h42}^{\prime}\\
		\bm{u}_{h51}^{\prime} & \bm{u}_{h52}^{\prime}
	\end{array}\right)=\left(\begin{array}{ccc}
		1 & 0 & 0\\
		0 & 0 & 1\\
		1 & 0 & 0\\
		0 & 0 & 1\\
		0 & 1 & 0
	\end{array}\right).
\end{align*}

\bs
Then, to relate our method to other methods, we can rewrite Equation $(\ref{eq:obj})$ as
\begin{gather}
\min_{\bm{U}, \bm{G}, \bm{B}} \phi(\bm{U}, \bm{G}, \bm{B}\,|\,\bm{Z},\bm{V})=\frac{1}{NHm}\sum_{j=1}^m\|\bm{U}\bm{G}-\bm{Z}_{j}^H\bm{B}_j\|^2\label{eq:obj2}\\
{\rm s.t.}\quad\frac{1}{Nm}\sum_{j=1}^m\bm{B}_j^{\prime}\bm{Z}_{j}^{H\prime}\bm{Z}^H_{j}\bm{B}_j=\bm{I}_p,\quad\bm{J}_{NH}\bm{U}\bm{G}=\bm{U}\bm{G}\nonumber\\
{\rm where}, \underset{(NH\times q_j)}{\bm{Z}_j^H}=\left(\begin{array}{c}
\bm{Z}_{j}\\
\vdots\\
\bm{Z}_{j}
\end{array}\right),\,\underset{(K\times p)}{\bm{G}}=\left(\begin{array}{c}
\bm{G}_{1}\\
\vdots\\
\bm{G}_{H}\\
\end{array}\right),\label{eq:UZG in obj2}
\end{gather}
\begin{gather*}
\bm{U}=\bdiag(\bm{U}_1, \bm{U}_2,\ldots, \bm{U}_H)=\left(\begin{array}{cccc}
\bm{U}_1 & \bm{0} & \cdots & \bm{0}\\
\bm{0}& \bm{U}_2 & \cdots & \bm{0}\\
\vdots & \vdots & \ddots & \vdots\\
\bm{0}& \bm{0} & \cdots & \bm{U}_H\\
\end{array}\right),
\end{gather*}
and $K=\sum_{h=1}^HK_{h}=\sum_{h=1}^H\sum_{s=1}^{r_s}K_{hs}$. If we set $H=1$ and define $\bm{U}_H$ as a cluster indicator matrix for $K_H$ clusters without the hierarchical clustering structure---that is, $\bm{U}_H=(u_{ik})$, $(i=1,\ldots,N\,;\,k=1,\ldots,K_H)$ where $\sum_{k=1}^{K_H}u_{ik}=1$ and $u_{ik}\in\{0,1\}$---then Equation $(\ref{eq:obj2})$ is equivalent to cluster CA  \cite{van2017cluster}, which is equivalent to GROUPALS \cite{van1989clusteringn} when applied to categorical variables. 

Thus, MSCCA represents an extension of cluster CA that is able to specify the cluster allocation for each class simultaneously in a common low-dimensional space, in which the quantifications for categories $\bm{B}_j$ $(j=1,\ldots,m)$ are optimally estimated for all clusters.

\subsection{Algorithm}
\label{sec:prop-algo}
To estimate the parameters $\bm{U},\bm{G}$, and $\bm{B}_j$ $(j=1,\ldots,m)$, we use an alternating least squares algorithm. 
The updating formulas come from a direct extension of cluster CA \cite{van2017cluster}.

\begin{description}
\item[Step 1: Initialization.] Determine $K_{hs}\,(h=1,\ldots,H\,;s=1,\ldots,r_h)$ and $p$. Set the number of iterations to $t=0$, and set a convergence criterion $\varepsilon$. Then, randomly generate initial clusters for each class.
\item[Step 2: Update $\bm{B}_j$.] Let $\bm{B}=(\bm{B}_1^{\prime},\ldots,\bm{B}_m^{\prime})^{\prime}$ and $\bm{Z}^H=(\bm{Z}_1^H,\ldots,\bm{Z}_m^H)$. Then find $\bm{B}^{(t+1)}$ as 
\begin{gather*}
\bm{B}^{(t+1)}=\sqrt{Nm}\bm{D}^{-1/2}\bm{B}^*\\
\where\quad\frac{1}{m}\bm{D}^{-1/2}\bm{Z}^{H\prime}\bm{J}_{NH}\bm{U}^{(t)}(\bm{U}^{(t)\prime}\bm{U}^{(t)})^{-1}\bm{U}^{(t)\prime}\bm{J}_{NH}\bm{Z}\bm{D}^{-1/2}=\bm{B}\bm{\Lambda}\bm{B}^{*\prime}\\
\bm{D}=\widetilde{\bm{Z}}^{\prime}\widetilde{\bm{Z}},\quad\widetilde{\bm{Z}}=\bdiag(\bm{Z}_1^H,\ldots,\bm{Z}_m^H)
\end{gather*}
\item[Step 3: Update $\bm{G}$.] Obtain $\bm{G}^{(t+1)}$ as follows:
\begin{gather*}
\bm{G}^{(t+1)}=\frac{1}{m}(\bm{U}^{(t)\prime}\bm{U}^{(t)})^{-1}\bm{U}^{(t)\prime}\bm{J}_{NH}\bm{ZB}^{(t+1)}
\end{gather*}
\item[Step 4: Update $\bm{U}$.] To obtain $\bm{U}_{h}^{(t+1)}$, the update proceeds by row. Specifically, each element in the $i$th row of $\bm{U}_{h}$, or $\bm{u}_{his}=(u_{hisk})$ $(k=1,\ldots,K_{hs})$, gets updated as follows: If $v_{his}=1$, 
\begin{gather*}
u_{hisk}^{(t+1)}=
\begin{cases}
1 &(k=\underset{\ell\in\{1,\ldots,K_{hs}\}}{\arg\min}\|\bm{f}_{i}-\bm{g}_{hs\ell}^{(t+1)}\|^2)\\
0 &({\rm others})
\end{cases}
\end{gather*}
and otherwise, $u_{hisk}^{(t+1)}=0$. Here, $\bm{f}_{i}$ is the $i$th row of $\bm{J}_{N}\bm{Z}\bm{B}^{(t+1)}$, and $\bm{g}_{hsk}^{(t+1)}$ is the cluster center of the $k$th cluster in the $s$th category in the $h$th supplementary variable.
\item[Step 5: Convergence test] Compute $\phi^{(t)}$, the value of the objective function from Equation $(\ref{eq:obj})$, using updated parameters. For $t>1$, if $\phi^{(t)}-\phi^{(t-1)}<\varepsilon$, terminate; otherwise, let $t=t+1$ and return to Step 2.
\end{description}

\subsection{Biplots}
\label{sec:prop-biplot}

In this subsection, we show how MSCCA can be used to construct a biplot. In \citeA{van2017cluster}, cluster CA is formulated as a maximization problem. Accordingly, the MSCCA in Equation ($\ref{eq:obj2}$) can be rewritten as the following maximization problem:
\begin{gather}
\max_{\bm{U},\bm{B}} \psi(\bm{U}, \bm{B}\,|\,\bm{Z}^H)=\tr\bm{B}^{\prime}\bm{Z}^{H\prime}\bm{J}_{NH}\bm{U}^{\prime}(\bm{U}^{\prime}\bm{U})^{-1}\bm{U}^{\prime}\bm{J}_{NH}\bm{Z}^H\bm{B}\label{eq:cca}\\
{\rm s.t.}\quad\frac{1}{NHm}\sum_{j=1}^m\bm{B}_j^{\prime}\bm{Z}_{j}^{H\prime}\bm{Z}_{j}^H\bm{B}_j=\bm{I}_p\nonumber
\end{gather}
The proof for the equivalence of Equations ($\ref{eq:obj2}$) and ($\ref{eq:cca}$) is in Proposition $\ref{prop:groupals and cca}$ of Appendix $\ref{sec:proof}$. When we leave $\bm{U}$ fixed, maximizing Equation ($\ref{eq:cca}$) is equivalent to minimizing
\begin{gather}
\min_{\bm{G}, \bm{B}} \phi^{CA}(\bm{G}, \bm{B}\,|\,\bm{Z}^H,\bm{V}, \bm{U})=\|\widetilde{\bm{P}}-\bm{D}_r^{1/2}\bm{G}\bm{B}^{\prime}\bm{D}_c^{1/2}\|^2\label{eq:DPDmin}\\
\s.t.\quad\frac{1}{Nm}\bm{B}^{\prime}\bm{D}_c\bm{B}=\bm{I}_p\nonumber\\
\where\quad\widetilde{\bm{P}}=\bm{D}_r^{-1/2}(\bm{P}-\bm{rc}^{\prime})\bm{D}_c^{-1/2}\label{eq:DPD}\\
\bm{P}=(Nm)^{-1}\bm{U}^{\prime}\bm{Z}^H,\,\,\,\bm{r}=\bm{P1}_{Q},\,\,\bm{c}=\bm{P}^{\prime}\bm{1}_{K},\quad\bm{D}_r=\diag(\bm{r}),\quad\bm{D}_c=\diag(\bm{c})\nonumber
\end{gather}
The proof of the equivalence is available from \citeA{van2017cluster}. Here, $\bm{P}$ indicates a $K\times Q$ scaled contingency table of clusters for each class (row) and category (column), and each element in $\bm{rc}^{\prime}$, $r_kc_{\ell}$ $(k=1,\ldots,K;\,\ell=1,\ldots,Q)$, indicates the scaled expected frequency with an assumption of independence between the $k$th cluster and the $\ell$th category. Thus, the matrix $\widetilde{\bm{P}}$ represents the standardized deviations from the the assumption of independence between cluster membership and the categorical variables. 

From Equation $(\ref{eq:DPDmin})$, it follows that the inner product of $\bm{D}_r^{1/2}\bm{G}$ and $\bm{D}_c^{1/2}\bm{B}$ approximates the matrix of standardized deviations from independence, $\widetilde{\bm{P}}$. That is, in MSCCA, we can use $\bm{G}$ and $\bm{B}$ to construct a biplot in which a greater the inner product of the $k$th row vector of $\bm{G}$ and the $\ell$th row vector in $\bm{B}$ generally indicates a stronger association between the $k$th cluster and the $\ell$th category. 

Note that in the resulting biplot, the points of the row and column are not necessarily similarly spread \cite<e.g.,>{gower2010area}. In this case, the points can be scaled using a constant, such that the average squared deviation from the origin of the row and column points is the same. See \citeA{van2017cluster} for detail.

%
\subsection{Relationship to the linear row constraint approach}
\label{sec:relation}
%

\citeA{hwang2002generalized} show that several approaches for incorporating external information about individuals into an MCA biplot can be generalized, as a linear row constraint framework. 

To add linear row constraints in MCA, we formulate the following objective function
\begin{gather}
\min_{\bm{G}, \bm{B}} \phi^{const}(\bm{G}, \bm{B}\,|\,\bm{Z}_j, \bm{V}_h)=\frac{1}{Nm}\sum_{j=1}^{m}\|\bm{C}\bm{F}-\bm{Z}_{j}\bm{B}_{j}\|^2\label{eq:obj const mca}\\
{\rm s.t.}\quad\frac{1}{Nm}\sum_{j=1}^{m}\bm{B}_{j}^{\prime}\bm{Z}_{j}^{\prime}\bm{Z}_{j}\bm{B}_{j}=\bm{I}_p,\quad\bm{J}_N\bm{CF}=\bm{CF},\nonumber
\end{gather}
where $\bm{C}$ is the $N\times N$ matrix that contains linear row constraints for the quantifications. If $\bm{C}=\bm{I}$, the problem reduces to the homogeneity formulation of MCA. 

The choice of $\bm{C}$ depends on the objective that underlies the incorporation of the external information. For example, if we were to use $\bm{C}=\bm{V}(\bm{V}^{\prime}\bm{V})^{-1}\bm{V}^{\prime}$, where $\bm{V}=\bdiag(\bm{V}_1,\ldots,\bm{V}_H)$, and we inserted $\bm{Z}_j^H$ for $\bm{Z}_j$, then Equation ($\ref{eq:obj const mca}$) would produce the averaging approach we described previously, because the class (category) would be represented by the average quantification of individuals corresponding to that class. Alternatively, if we aimed to ``remove" the effect of external information from a biplot, then we might use $\bm{C}=\bm{I}-\bm{V}(\bm{V}^{\prime}\bm{V})^{-1}\bm{V}^{\prime}$ \cite<e.g.,>{takane1991principal,takane2002generalized,hwang2002generalized}, which is equivalent to deducting the class conditional means from the data. For example, if as supplementary variable we have gender, the mean of all males is deducted from all male observations.

Although MSCCA follows a different approach from these two examples to incorporate external information, we can reformulate this method to fit into the linear row constraint framework. In particular, for a fixed $\bm{U}$, the MSCCA objective function in Equation $(\ref{eq:obj2})$ can be rewritten as a minimization problem:
\begin{gather}
\min_{\bm{G}, \bm{B}} \phi^{MSCCA}(\bm{G}, \bm{B}\,|\,\bm{Z}, \bm{U}, \bm{V})=\frac{1}{NHm}\sum_{j=1}^{m}\|\bm{C}\bm{F}-\bm{Z}_{j}^H\bm{B}_{j}\|^2\label{eq: mscca lin cons}\\
{\rm s.t.}\quad\frac{1}{Nm}\sum_{j=1}^{m}\bm{B}_{j}^{\prime}\bm{Z}_{j}^{H\prime}\bm{Z}_{j}^H\bm{B}_{j}=\bm{I}_p,\quad\bm{J}_{NH}\bm{CF}=\bm{CF}\nonumber\\
{\rm where}\quad\bm{C}=\bm{U}(\bm{U}^{\prime}\bm{U})^{-1}\bm{U}^{\prime}\nonumber\\
\underset{(NH\times K)}{\bm{U}}=\bdiag(\bm{U}_1, \bm{U}_2,\ldots, \bm{U}_H),\,\underset{(NH\times q_j)}{\bm{Z}_j^H}=\left(\begin{array}{c}
\bm{Z}_{j}\\
\vdots\\
\bm{Z}_{j}
\end{array}\right),\,\underset{(K\times p)}{\bm{G}}=\left(\begin{array}{c}
\bm{G}_{1}\\
\vdots\\
\bm{G}_{H}\\
\end{array}\right),\nonumber
\end{gather}
where, $\bm{U}$ still features the hierarchical cluster structure constraint imposed by Equation ($\ref{eq:constU}$). From this formulation, it immediately follows that  MSCCA represents a special case of Equation $(\ref{eq:obj const mca})$, with $\bm{C}=\bm{U}(\bm{U}^{\prime}\bm{U})^{-1}\bm{U}^{\prime}$. Proposition $\ref{prop:CCA and linear const}$ in Appendix $\ref{sec:proof}$ offers a proof of the equivalence of Equations $(\ref{eq:obj const mca})$ and $(\ref{eq: mscca lin cons})$.

%
\section{Numerical illustration of an MSCCA biplot}
\label{sec:numex}
%

\begin{table}[t]
	\centering
	\caption{{\small Categories of variables of artificial data for the simple illustration}}
	\bs
	\scalebox{0.9}{ 
		\begin{tabular}{lll}
			\hline
			Variable type & Variable name & Category\\
			\hline
			Variables to estimate quantifications & Meal & Western, Asian \\
			& Drink & Fruits juice, Tea, Alcohol\\
			Supplementary variables & Nationality & American, Japense \\
			& Gender & Male, Female\\
			\hline
		\end{tabular}
	}
	\label{tab:catename of sim}
\end{table}

In this section, we present a small example, using artifical data, to illustrate how MSCCA works. With this example, we zoom in specifically on the differences between MSCCA and the averaging approach for the visualization of heterogeneous tendencies. 

To start, we generate categorical data for 200 individuals that represent two categorical variables (meal and drink preference), and two supplementary variables (nationality and gender). Table $\ref{tab:catename of sim}$ contains the variables and corresponding categories. With this analysis, we seek to determine if different tendencies, with respect to the meal and drink preferences, emerge for groups of individuals, depending on their nationality and gender.

We generated the data to establish three true clusters in the full data set. Individuals in the first cluster choose ``Western meal" for the meal variable, and ``fruit juice" for the drink variable (W\&J), those in the second cluster choose ``Asian meal" and ``Tea" (A\&T), and in the third cluster, individuals choose ``Western meal" and ``alcohol" (W\&A). The frequency distribution of the generated artificial data over each cluster in each class is shown in Figure $\ref{fig:bar}$, revealing there are two clusters for Americans, Asians and females and three clusters for males.

\begin{figure}
	\begin{center}
		\includegraphics[width=12cm]{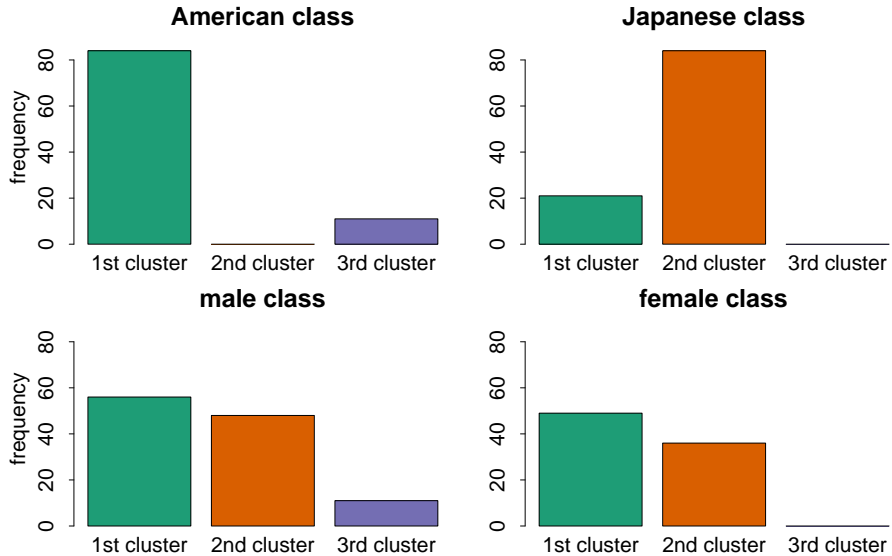}
		\caption{{\small Frequency distributions for each combination of supplementary categories for each true cluster}} 
		\label{fig:bar}
	\end{center}
\end{figure}

The biplot that results from an averaging approach, in Figure $\ref{fig:plot numex}$ (left), clearly reveals overall tendencies of many Americans and Japanese consumers, strongly associated with W\&J and A\&T, respectively. However, the much smaller number of individuals who choose ``alcohol" makes it impossible to specify who (i.e., which nationality or gender) makes this choice. 

In contrast, by obtaining clusters for each class, the MSCCA biplot makes the tendencies of this relatively small number of individuals visible. When we use the correct number of clusters for each class, the MSCCA biplot result in Figure $\ref{fig:plot numex}$ (right) clearly reveals that a small number of male Americans choose ``alcohol". In addition, this biplot still depicts the tendencies of the larger groups, as obtained in the averaging approach. That is, MSCCA reveals the tendencies of small groups, without losing the information about the tendencies of larger groups.

In Appendix $\ref{sec:how it works}$, we use a CA framework to provide some additional insights into how (and when) the MSCCA and averaging approaches differ with respect to their depiction of heterogeneous tendencies.

\begin{figure}
	\begin{center}
		\includegraphics[width=14cm]{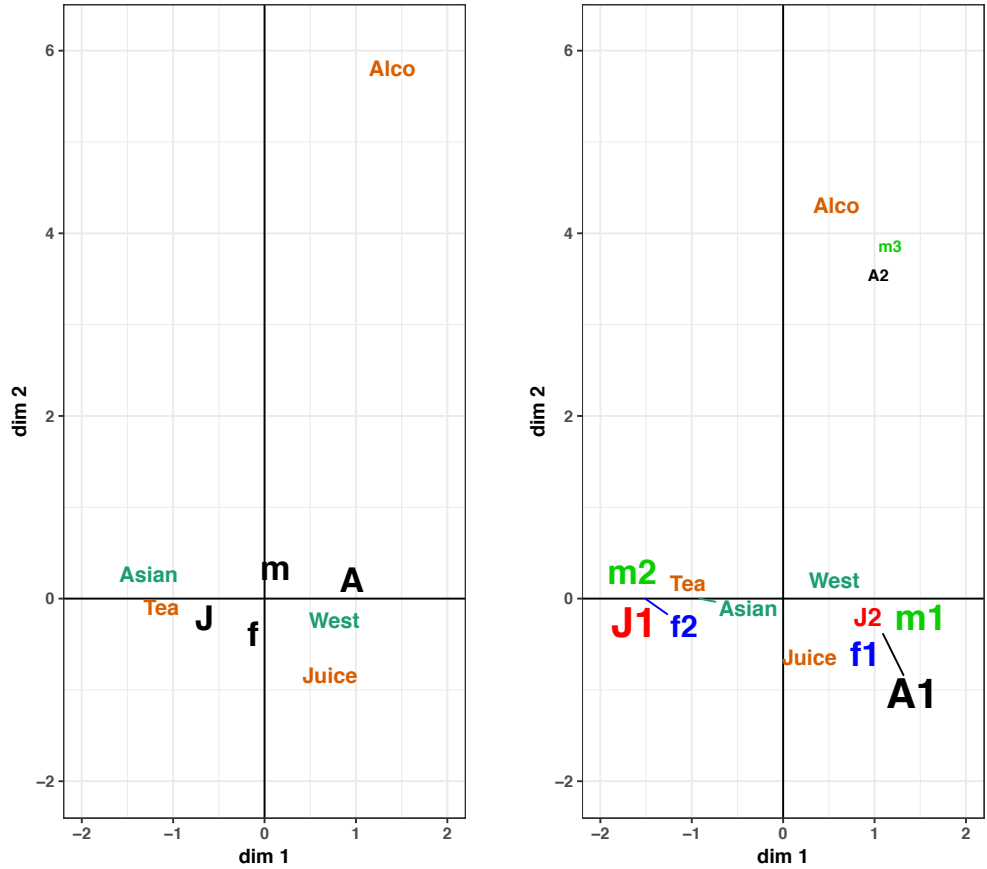}
		\caption{{\small Results (Left) Averaging approach. Average points are labelled ``A" (American), ``J" (Japanese), ``m" (male), and ``f" (female), and the label sizes correspond to class sizes. Other character labels indicate category points. (Right) MSCCA. The points labeled ``A," ``J," ``m," or ``f" followed by a number, correspond to the cluster points for each class.}}
		\label{fig:plot numex}
	\end{center}
\end{figure}

%
\section{Simulation}
\label{sec:sim}
%

We conducted a simulation study to evaluate the performance of MSCCA in different scenarios. By using simulations, we can determine the effects of the supplementary variables on the accuracy of the clustering and biplots achieved through MSCCA.

\subsection{Data Generation}
\label{sec:sim-data}

The data generation process consists of two steps: generating an $N \times m$ data matrix, and generating $N \times H$ matrix of supplementary variables. First, we start by dividing the $m$ variables into two groups: active variables that relate to the clustering structure, and noise variables that are unrelated to the cluster structure. Furthermore, we determine the cluster allocation with a multinomial distribution. To generate data for the active variables, we assign one category for each variable a high probability of $0.8$. Then the (low) probabilities assigned to the remaining categories are determined according to $\bar{\bm{p}}=(\bar{p}_{\ell})$ $(\ell=1,\ldots,q-1)$, where $\bar{\bm{p}}=((1-0.8)\times(p_1,\ldots,p_{q-1})/\sum_{\ell=1}^{q-1}p_{\ell})$ and $p_{\ell}\sim U(0, (1-0.8))$. The high probability categories are cluster specific. Then to generate noise variables, we use a multinomial distribution in which the proportion for each category is $1/q$. In our simulation study, we set the ratio of active to noise variables to $1:1$. 

Second, to generate the data matrix corresponding to the $H$ supplementary variables, we consider two scenarios: balanced and unbalanced distributions over the categories. In the balanced scenarios, the multinomial probabilities for all categories are equal. In the unbalanced scenario, the probabilities are $1/S, \ldots, r_h/S$, where $r_h$ denotes the number of categories for the supplementary variable, and $S=\sum_{s=1}^{r_h}s$.

\subsection{Simulation study design}
\label{sec:sim-design}
To assess the performance of the methods in different settings, we fix the number of observations $N = 300$ and the number of variables $m = 10$. Then, we consider a full factorial design with the number of categories for each variable $q = 5, 7$; the number of clusters $K = 2,3$; the number of supplementary variables $H = 1,3$; and the number of categories for the supplementary variables $r_h = 3, 5$. Finally, for the supplementary variables we note the balanced and unbalanced scenarios. For each combination of parameters in the simulation, we randomly generate 100 different $N\times m$ data matrices and $N\times H$ supplementary variable matrices. For each data set, we apply MSCCA using 100 random initial values.

\subsection{Evaluation}
\label{sec:sim-eval}
We evaluate the performance of the proposed methods by checking the accuracy of both the clustering and the biplots. To measure clustering accuracy, we turn to the Adjusted Rand Index (ARI, \citeNP{hubert1985comparing}). The ARI assesses the similarity between two cluster allocations, so it takes a value of 1 for a perfect recovery, and this value decreases as performance worsens. We calculate the ARI for the class-specific clustering results separately.

For biplot accuracy, we use a goodness-of-fit (GF) index \cite<e.g.,>{gabriel2002goodness}, which is equivalent to the so-called congruence coefficient \cite<e.g.,>{lorenzo2006tucker}. The GF between configurations $\bm{Y}$ and $\bm{H}$ is defined as
\begin{align*}
	{\rm GF}(\bm{Y},\bm{H})=\frac{\tr^2(\bm{Y}^{\prime}\bm{H})}{\tr(\bm{Y}^{\prime}\bm{Y})\tr(\bm{H}^{\prime}\bm{H})}=\cos^2(\bm{Y},\bm{H}).
\end{align*}
Therefore, we calculate the GF between $\bm{Y}$ and $\bm{H}$, where $\bm{H}=\bm{GB}^{\prime}$ (with $\bm{G}$ and $\bm{B}$ as the MSCCA solutions) and $\bm{Y}=\tilde{\bm{P}}^{true}=\bm{D}_r(\bm{P}^{true}-\bm{rc}^{\prime})\bm{D}_c$, such that $\bm{P}^{true}=\bm{U}^{\prime}\bm{Z}$ and $\bm{U}$ is the true cluster allocation. Note that by definition, GF$\in[0,1]$. In our calculation of the GF index, we assume that the true cluster allocation is known. Therefore, cluster accuracy does not affect the GF index.  

\subsection{Result}
\label{sec:sim-res}
\begin{figure}
	\begin{center}
		\includegraphics[width=10cm, angle=270]{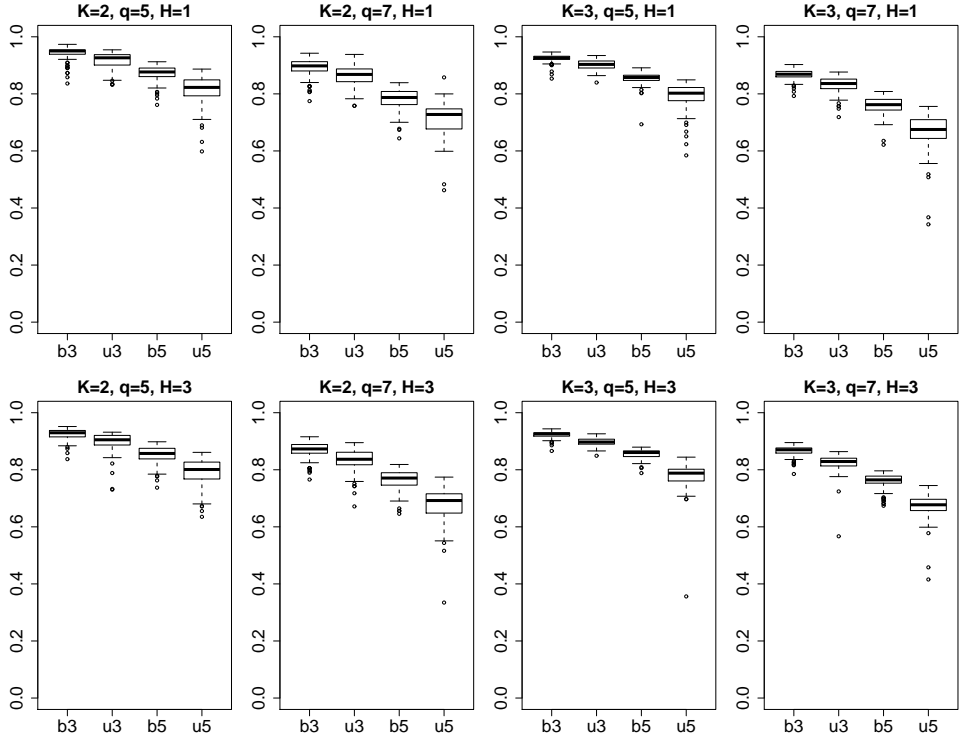}
		\caption{{\small Boxplot of GF for each case. On the horizontal axis, b3 indicates balanced categories in supplementary variables, and $r_h=3$ for all $h=1,\ldots,H$; while b5 has balanced categories with $r_h=5$ for all $h$. Similarly, u3 and u5 indicate unbalanced categories with $r_h=3,5$, respectively.}}
		\label{fig:simres_GF}				
	\end{center}
\end{figure}

\begin{figure}
	\begin{center}
		\includegraphics[width=10cm, angle=270]{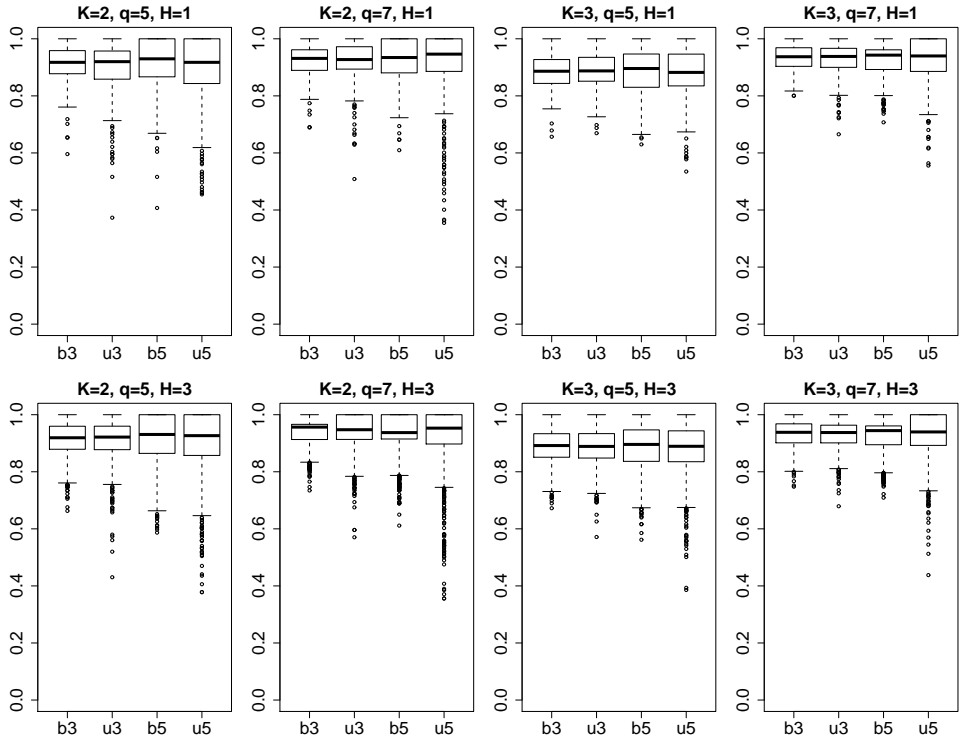}
		\caption{{\small Boxplot of ARI for each case.}}
		\label{fig:simres_ARI}
	\end{center}
\end{figure}

The results for the GF index in Figure $\ref{fig:simres_GF}$ indicate that it tends to decrease as the number of categories $q$ increases. The number of supplementary variables $H$ does not substantially affect the GF. Rather, the GF tends to be somewhat better when there are fewer categories $r_h$ in the supplementary variables and when the distribution over the categories is balanced.

The cluster retrieval results in Figure $\ref{fig:simres_ARI}$ show that overall, ARI decreases when the number of clusters $K$ increases and when the number of categories $q$ decreases. In contrast, the number of supplementary variables $H$ and whether the distributions over the categories are unbalanced do not affect the median ARI substantially. However, for more supplementary variables with balanced distributions, we note more outlying results. In addition, the number of categories for the supplementary variables appears to affect variance in the ARI results, such that the ARI for $r_h = 5$ has greater variance than that for $r_h = 3$.

\subsection{Conclusions from the simulation study}
\label{sec:sim-conclu}

The simulation study shows that the number of supplementary variables does not affect the accuracy of the biplot or the clustering. We can increase the number of supplementary variables $H$ without harming the accuracy of the results. However, increasing the number of supplementary variables $H$ leads to more points in the biplot, resulting in a more complicated visualization. We thus assert that $H$ can be increased as long as the biplot remains interpretable.

In addition, though the clustering results are hardly affected by the nature of the supplementary variables (i.e., number of categories $r_h$, and whether the distribution over the categories is balanced), the simulation study indicates that biplot accuracy is affected. In particular, using supplementary variables with more categories and unbalanced distributions over categories leads to a decrease in biplot accuracy. In conclusion, when there are several candidates for supplementary variables, it is better to select balanced supplementary variables with fewer categories.

%
\section{Application}
\label{sec:real}
%

In this section, we illustrate the proposed method using data that reflect road accidents in the United Kingdom. With these data, we seek to determine how the circumstances in which a car accident occurs depends on the type of accident. We compare the results using MSCCA, the averaging approach, and cluster CA, to establish how each method would visualize the relationships.

\subsection{Data and Setting}
\label{sec:real-data}

\begin{table}
	\centering
	\caption{{\small Categories for each variable and their corresponding labels in biplots and descriptions.}}
	\bs
	\scalebox{0.8}{ 
		\begin{tabular}{llll}
			\hline
			Variable type & Variable name & Label & Description\\
			\hline
			Non-supplementary variables & Light conditions & Dark0 & Daylight\\
			&  & Dark1 & Darkness: street lights present and lit \\
			&  & Dark2 & Darkness: street lights present but unlit \\
			&  & Dark3 & Darkness: no street lighting \\
			& Weather conditions & Fine & Fine without high winds \\
			&  & Rain & Raining without high winds \\
			&  & Snow & Snowing without high winds \\
			&  & Fine\_w & Fine with high winds \\
			&  & Rain\_w & Raining with high winds \\
			&  & Snow\_w & Snowing with high winds \\
			&  & Fog & Fog or mist — if hazard \\
			&  & Other & Other \\
			& Road surface conditions & Dry & Dry \\
			&  & Wet & Wet / Damp \\
			&  & Snow & Snow \\
			&  & Frost & Frost / Ice \\
			&  & Flood & Flood (surface water over 3cm deep) \\
			& Speed Limit & $\sim$30 & Speed limit is up to 30km/h\\
			&  & $\sim$70 & Speed limit is up to 70km/h\\
			Supplementary variables & Casualty class & Driver & Casualty is one driver\\
			&  & Ped & Casualty is one pedestrian\\
			& Area & Urban & Occurring in urban area\\
			&  & Rural & Occurring in rural area\\
			\hline
		\end{tabular}
	}
	\label{tab:catename}
\end{table}

The data were obtained from the U.K. Department for Transport's road safety statistics (https://data.gov.uk/dataset/road-accidents-safety-data). In these data, observations are accidents, and the (categorical) variables refer to information about those accidents. For this illustration, we selected accidents that occurred in January 2016, that involved one casualty (either a driver or a pedestrian), and in which at most two parties were involved. The resulting data set contains $N = 3,026$ observations.
 
Regarding the circumstances of the accident, we consider four (i.e., $m = 4$) variables: lighting conditions, weather conditions, road surface conditions, and speed limit. For the types of accident, we select two ($H = 2$) supplementary variables: casualty class and area. Table $\ref{tab:catename}$ summarizes the variables and their categories. 

As is true of any cluster analysis method, determining the number of clusters is not trivial. In MSCCA, the number of clusters must be prespecified for each class $K_{hs}$ ($h = 1,\ldots, H ; s = 1,\ldots, r_h$). For this study, we use the Krzanowski-Lai index (KL index, \citeNP{krzanowski1988criterion}) to determine the number of clusters for each class, with separate cluster CA analyses. Specifically, we apply cluster CA to class-specific data (i.e., data corresponding to one category of the supplementary variables) to determine the number of clusters $K_{hs}$ that corresponds to the optimal KL index. This procedure results in four clusters for the driver class, five clusters for the pedestrian class, four in the urban class and four clusters for the rural class (i.e., $K_{11} = 4$, $K_{12} = 5$, $K_{21} = 4$ and $K_{22} = 4$). Henceforward, we refer to a cluster from the driver class as driver cluster, clusters from the pedestrian class as pedestrian clusters, and so on.

In a comparative analysis, we also consider the averaging approach and cluster CA with complete data (i.e., including the supplementary variables in the analysis to determine clusters and quantifications). To select the number of clusters for the complete cluster CA analysis, we employed the KL index and obtained $K = 7$ clusters.

\subsection{Result}
\label{sec:real-res}

\subsubsection{MSCCA result}

\begin{figure}
	\begin{center}
		\includegraphics[width=16cm]{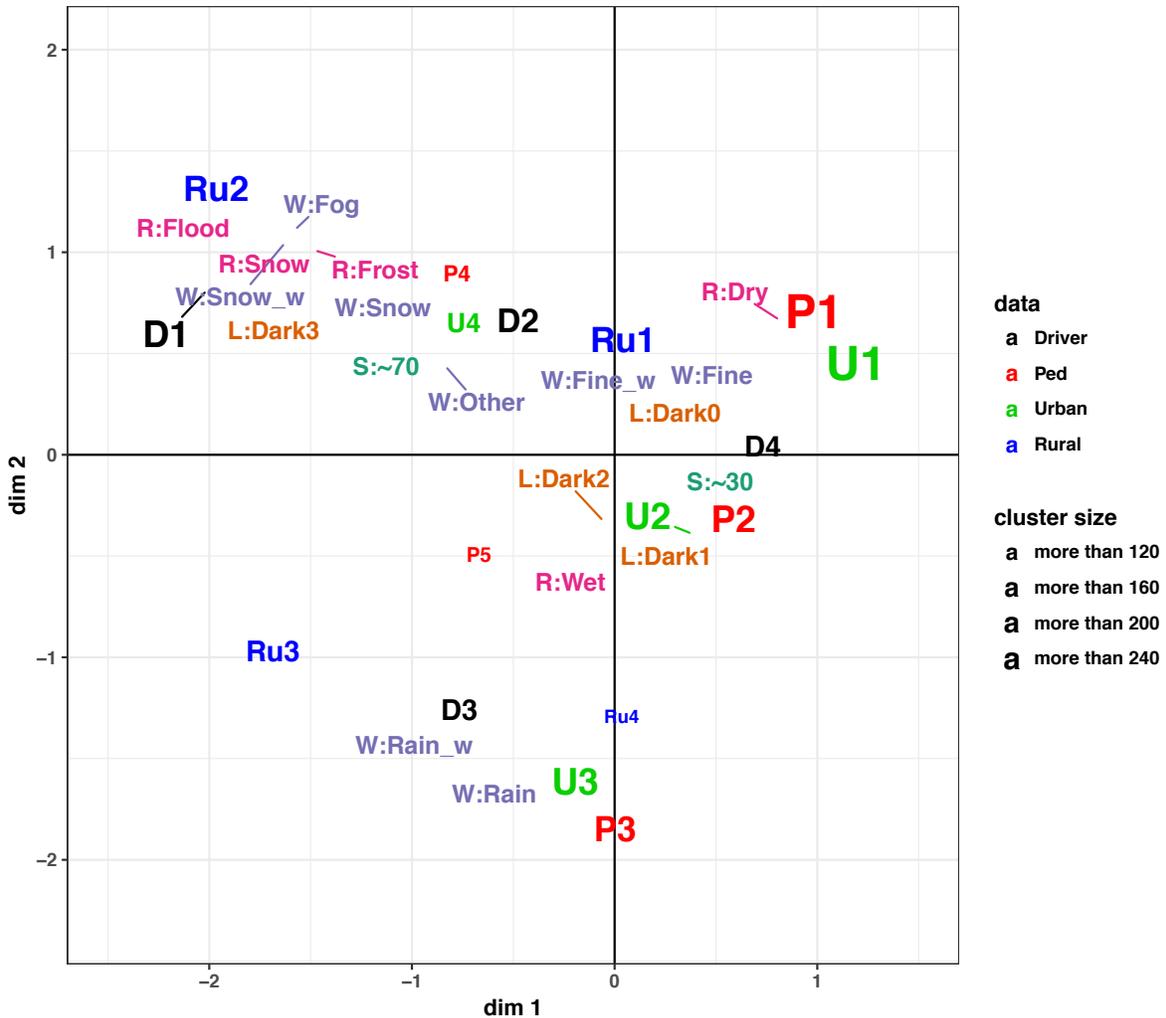}
	\end{center}
	\caption{{\small Results using MSCCA. The numbered labels indicate cluster points with ``D" indicating driver clusters, ``P" corresponding to the pedestrian class, ``U" to the urban class, and ``Ru" to rural class clusters. The numbers reflect the size of each cluster within its class (e.g., ``D1" indicates the largest size cluster in the driver class), as also indicated by the label sizes. Character labels also indicate light conditions ``L", weather conditions ``W", road surface conditions ``R", and speed limits ``S". }}
	\label{fig:real res prop}
\end{figure}

In the biplot for the MSCCA solution (Figure $\ref{fig:real res prop}$), we see that the largest pedestrian clusters, as well as the largest urban and rural clusters (P1, U1, and R1, respectively) are related to categories such as ``Fine," ``Fine\_w," and ``Dry." That is, many accidents in urban and rural areas result in pedestrian casualties and have a strong association with what is generally be considered good driving conditions (e.g., fine weather, dry roads). 

The driver cluster (D1) instead is related to categories such as ``Dark3," ``Snow\_w (weather condition)," and ``Snow (road surface)." Therefore, many accidents that result in driver casualties have a strong association with bad driving conditions, such as a dark night or slippery road. Another driver cluster, close to the good conditions, is the smallest one, indicating that accidents resulting in a driver casualty are less likely under good driving conditions.

In the rural class, we also recognize that though the largest rural cluster is proximal to categories that correspond to good conditions, the second largest rural cluster is close to bad conditions. Therefore, accidents in rural areas occur in both good and bad driving conditions. 

The fourth-largest cluster for rural data and the third-largest clusters for the three other classes indicate similar associations with categories such as ``Rain," ``Rain w," and (to some extent) ``Wet." This indicates that for all classes of supplementary variables, some clusters of accidents occur in rainy weather.

By inspecting the MSCCA biplot and relating the class-specific cluster points to the category quantifications, we can visually perceive how accidents, split into different classes, relate differently to weather and road conditions. For example, for pedestrians, the risk of casualties exists even in favorable conditions, but accidents involving drivers are more strongly related to bad conditions.

\subsubsection{Averaging approach results}

\begin{figure}
	\begin{center}
		\includegraphics[width=12cm]{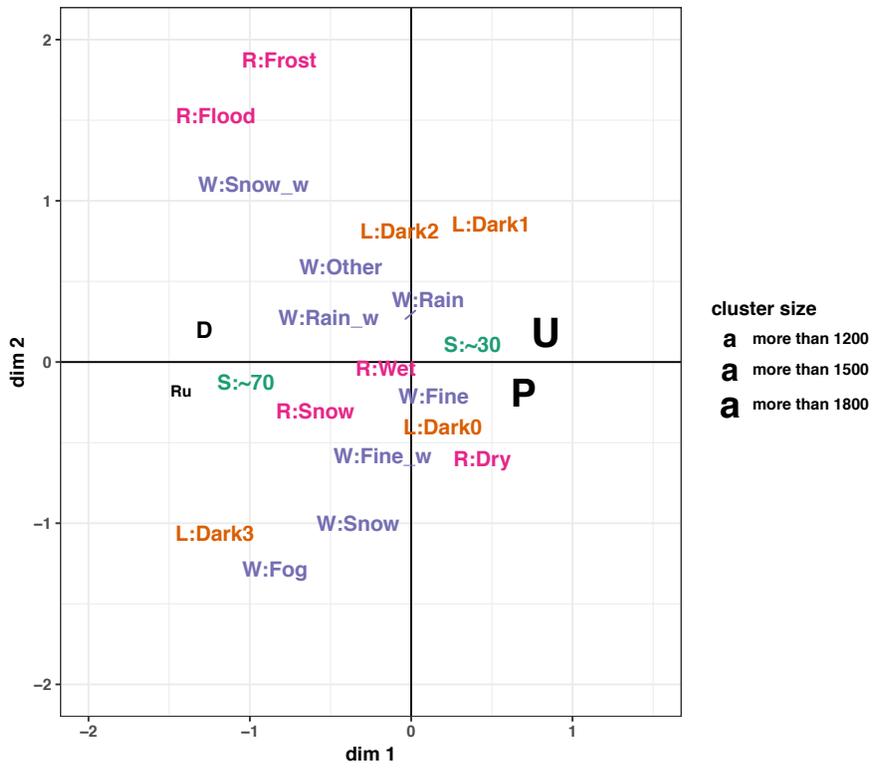}
	\end{center}
	\caption{{\small Results using the averaging approach. The character labels ``D", ``P", ``U" and ``Ru" indicate classes defined by supplementary variables, label sizes correspond to class sizes. Other character labels indicate category points, same as Figure $\ref{fig:real res prop}$.}}
	\label{fig:real res ave}
\end{figure}

The results using the averaging approach are in Figure $\ref{fig:real res ave}$. We can still interpret the information regarding classes with respect to categories, but the averaging of the results limits the available information. Specifically, we see that ``Driver" and ``Rural" relate to categories indicating bad driving conditions (e.g., ``$\sim$70", ``Show"), while ``Pedestrian" and ``Urban" are related to categories corresponding to good driving conditions (e.g., ``$\sim$30", ``Fine", ``Dark0"). However, it is difficult to interpret the relationship between classes and categories that are not close to the class quantifications. Averaging limits us to interpreting tendencies that many accidents in each class have in common. Differentiation with respect to smaller, relatively homogeneous subgroups is no longer possible.

\subsubsection{Cluster CA result}

\begin{figure}
	\begin{center}
		\includegraphics[width=12cm]{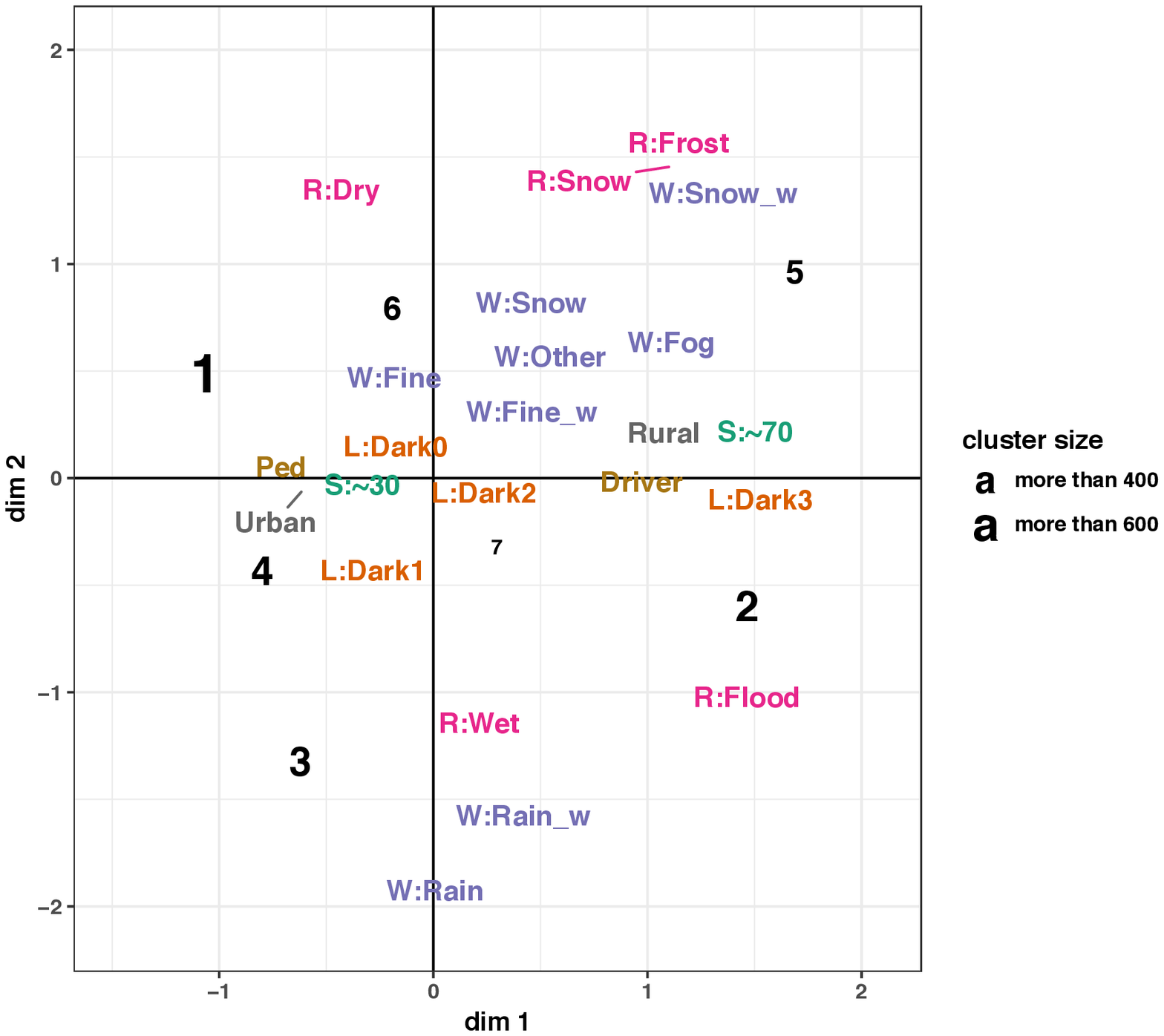}
	\end{center}
	\caption{{\small Results using cluster CA. Numbered labels indicate cluster points and the numbers are ordered according to the sizes of the clusters. The other labels are the same as in Figures $\ref{fig:real res prop}$ and $\ref{fig:real res ave}$.}}
	\label{fig:real res cca}
\end{figure}

Figure $\ref{fig:real res cca}$ shows the results using the cluster CA approach. In contrast with the averaging approach, we can now distinguish different clusters corresponding to several accident tendencies. For example, we find a cluster associated with rain-related categories, whereas this relationship was not clear in the averaging results. Yet the cluster CA approach still limits interpretations with respect to classes. For example, we can see that ``Pedestrian" and ``Urban" are related to good driving conditions, but we cannot interpret the relationship between the ``Pedestrian" and ``Urban" class in conditions such as rainy or bad driving conditions (e.g, ``$\sim$70" and ``Dark3"). In contrast, with MSCCA, we can better interpret these relationships (e.g., we can see that the ``Pedestrian" class has a weaker association with bad driving conditions than with good ones or with rainy conditions, because the smallest pedestrian cluster is closest to bad driving conditions.)

\subsection{Conclusions of Application}

In this section, we have compared three visualization results to appraise differences in how the biplots incorporate external information. All three methods can identify situations in which many accidents occur in each class. However, only by using MSCCA were we able to differentiate across conditions in which many or few accidents occurred. Specifically, this method reveals that relatively many accidents in the ``Pedestrian" and ``Urban" classes occur when conditions are good, but fewer occur when conditions are bad. Conversely, for the ``Driver" class, accidents predominantly occur under bad conditions, with only a few appearing when conditions are good. For accidents corresponding to the ``Rural" class, we find that they occur in both good and bad conditions. Finally, for all classes, we uncover relatively small clusters of accidents that relate strongly to rainy conditions.

%
\section{Conclusion}
\label{sec:discus}
%

We have proposed a new approach to incorporate and interpret external information in a biplot for categorical data. Specifically, we introduce a multiple-set extension to cluster CA, MSCCA, that can visually establish the relationship between external information and categories. In MSCCA, the class-specific clusters obtained make it possible to identify heterogeneous tendencies within each class. We also show how MSCCA relates to a linear row constraint framework.

To investigate the performance of this proposed method, we consider different conditions, according to a simulation study. The results show that increasing the number of supplementary variables $H$ has little effect on cluster or biplot accuracies. However, the results are better if the supplementary variables feature few categories and a balanced distribution over categories.

Then with an empirical analysis of road accident data, we show that that the averaging and cluster CA approaches can uncover only tendencies corresponding to the majority of accidents in each class. The MSCCA biplot instead makes it possible to interpret heterogeneous tendencies within each class, regardless of cluster sizes.

Finally, MSCCA can be applied to different settings. In particular, it could be adopted in a three-way setting to depict the relationship among multiple two-way data sets. For example, if we have $N \times m$ categorical data sets corresponding to $T$ different time points, we could use MSCCA to reveal the relationships among clusters at different times.

\nocite{greenacre2013contribution}


\newpage
\begin{appendices}

%
\section{Proof}
\label{sec:proof}
%
In this Appendix, we consider several propositions regarding MSCCA. Because MSCCA is an extension of cluster CA, without loss of generality, we provide the proof for cluster CA.

\begin{prop}\label{prop:groupals and cca}
	Two optimization problems
	\begin{gather}
	\min\phi(\bm{U}, \bm{G}, \bm{B}\,|\,\bm{Z})=\frac{1}{Nm}\sum_{j=1}^m\|\bm{U}\bm{G}-\bm{Z}_{j}\bm{B}_j\|^2\label{eq:ap grpals}\\
	{\rm s.t.}\quad\frac{1}{Nm}\sum_{j=1}^m\bm{B}_j^{\prime}\bm{Z}_{j}^{\prime}\bm{Z}_{j}\bm{B}_j=\bm{I}_p,\quad\bm{J}_{N}\bm{U}\bm{G}=\bm{U}\bm{G}\nonumber
	\end{gather}
	and 
	\begin{gather}
	\max\psi(\bm{U}, \bm{B}\,|\,\bm{Z})=\tr\bm{B}^{\prime}\bm{Z}^{\prime}\bm{J}_N\bm{U}^{\prime}(\bm{U}^{\prime}\bm{U})^{-1}\bm{U}^{\prime}\bm{J}_N\bm{ZB}\label{eq:ap cca}\\
	{\rm s.t.}\quad\frac{1}{Nm}\sum_{j=1}^m\bm{B}_j^{\prime}\bm{Z}_{j}^{\prime}\bm{Z}_{j}\bm{B}_j=\bm{I}_p\nonumber
	\end{gather}
	are equivalent.
\end{prop}

\begin{proof}
	At first the equivalence is shown when $\bm{U}$ is fixed. Considering the constraints, Equation $(\ref{eq:ap grpals})$ can be rewritten as
	\begin{align*}
	\phi&=\frac{1}{Nm}\sum_{j=1}^m\|\bm{U}\bm{G}-\bm{Z}_{j}\bm{B}_j\|^2\\
	&=\frac{1}{Nm}\left(m\tr\bm{G}^{\prime}\bm{U}^{\prime}\bm{UG}-2\tr\sum_{j=1}^m\bm{B}_j^{\prime}\bm{Z}_j^{\prime}\bm{UG}+\tr\sum_{j=1}^m\bm{B}_j^{\prime}\bm{Z}_j^{\prime}\bm{Z}_j\bm{B}_j\right)\\
	&=\frac{1}{N}\tr\bm{G}^{\prime}\bm{U}^{\prime}\bm{UG}-\frac{2}{Nm}\tr\bm{B}^{\prime}\bm{Z}^{\prime}\bm{UG}+p,&
	\end{align*}
	because $\frac{1}{Nm}\sum_{j=1}^m\bm{B}_j^{\prime}\bm{Z}_j^{\prime}\bm{Z}_j\bm{B}_j=\bm{I}_p$. Using $\bm{J}_{N}\bm{U}\bm{G}=\bm{U}\bm{G}$ and omitting the constant, this minimization will be
	\begin{align}
	\frac{1}{N}\tr\bm{G}^{\prime}\bm{U}^{\prime}\bm{UG}-\frac{2}{Nm}\tr\bm{B}^{\prime}\bm{Z}^{\prime}\bm{J}_N\bm{UG}\label{eq:ap prf1}
	\end{align}
	Solving this for $\bm{G}$, we obtain
	\begin{gather*}
	\bm{G}=\frac{1}{m}(\bm{U}^{\prime}\bm{U})^{-1}\bm{U}^{\prime}\bm{J}_N\bm{ZB}
	\end{gather*}
	Inserting this in Equation ($\ref{eq:ap prf1}$), it will be
	\begin{align*}
	&\frac{1}{Nm^2}\tr\bm{B}^{\prime}\bm{Z}^{\prime}\bm{J}_N\bm{U}^{\prime}(\bm{U}^{\prime}\bm{U})^{-1}\bm{U}^{\prime}\bm{J}_N\bm{ZB}-\frac{2}{Nm^2}\tr\bm{B}^{\prime}\bm{Z}^{\prime}\bm{J}_N\bm{U}^{\prime}(\bm{U}^{\prime}\bm{U})^{-1}\bm{U}^{\prime}\bm{J}_N\bm{ZB}\\
	&=-\frac{1}{Nm^2}\tr\bm{B}^{\prime}\bm{Z}^{\prime}\bm{J}_N\bm{U}^{\prime}(\bm{U}^{\prime}\bm{U})^{-1}\bm{U}^{\prime}\bm{J}_N\bm{ZB}
	\end{align*}
	Minimizing this is equivalent to maximizing Equation ($\ref{eq:ap cca}$).
	Next, the equivalence is shown when $\bm{B}$ is fixed and $\bm{U}$ is not. At first, a \textit{k}-means type optimization problem
	\begin{align*}
	\min_{\bm{U},\bm{G}}\|\bm{UG}-\bm{J}_N\bm{ZB}\|^2
	\end{align*}
	is equivalent to the optimization problem in Equation ($\ref{eq:ap cca}$), since this can be rewritten as
	\begin{align}
	\|\bm{UG}-\bm{J}_N\bm{ZB}\|^2&=\tr\bm{G}^{\prime}\bm{U}^{\prime}\bm{UG}-2\tr\bm{B}^{\prime}\bm{Z}^{\prime}\bm{J}_N\bm{UG}+\tr\bm{B}^{\prime}\bm{Z}^{\prime}\bm{J}_N\bm{ZB}\nonumber\\
	&=-\tr\bm{B}^{\prime}\bm{Z}^{\prime}\bm{J}_N\bm{U}^{\prime}(\bm{U}^{\prime}\bm{U})^{-1}\bm{U}^{\prime}\bm{J}_N\bm{ZB}+\tr\bm{B}^{\prime}\bm{Z}^{\prime}\bm{J}_N\bm{ZB}\label{eq:grpals cca U 1}
	\end{align}
	Here, we use $\bm{G}=m^{-1}(\bm{U}^{\prime}\bm{U})^{-1}\bm{U}^{\prime}\bm{J}_N\bm{ZB}$. Omitting a constant term, minimizing Equation ($\ref{eq:grpals cca U 1}$) is equivalent to maximizing Equation ($\ref{eq:ap cca}$). On the other hand, with $\bm{B}_j$ ($j=1,\ldots,m$) fixed, Equation ($\ref{eq:ap grpals}$) can be written as
	\begin{align*}
	\sum_{j=1}^m\|\bm{U}\bm{G}-\bm{Z}_{j}\bm{B}_j\|^2&=\|\bm{U}\bm{G}-\bm{Z}\bm{B}\|^2\\
	&=\tr\bm{G}^{\prime}\bm{U}^{\prime}\bm{UG}-2\tr\bm{B}^{\prime}\bm{Z}^{\prime}\bm{J}_N\bm{UG}+\tr\bm{B}^{\prime}\bm{Z}^{\prime}\bm{ZB}\\
	&=-\tr\bm{B}^{\prime}\bm{Z}^{\prime}\bm{J}_N\bm{U}^{\prime}(\bm{U}^{\prime}\bm{U})^{-1}\bm{U}^{\prime}\bm{J}_N\bm{ZB}+\tr\bm{B}^{\prime}\bm{Z}^{\prime}\bm{ZB}
	\end{align*}
	This is the same as Equation ($\ref{eq:grpals cca U 1}$). Thus, we obtain the proposition.
\end{proof}

\begin{prop}\label{prop:CCA and linear const}
	Minimizing Equation $(\ref{eq:ap grpals})$ with respect to $\bm{B}$ is equivalent to minimizing
	\begin{gather}
	\phi^{const}(\bm{B}\,|\,\bm{Z},\bm{U})=\frac{1}{Nm}\sum_{j=1}^{m}\|\bm{C}\bm{F}-\bm{Z}_{j}\bm{B}_{j}\|^2\label{eq:ap linear const}\\
	{\rm s.t.}\,\,\,\frac{1}{Nm}\sum_{j=1}^{m}\bm{B}_{j}^{\prime}\bm{Z}_{j}^{\prime}\bm{Z}_{j}\bm{B}_{j}=\bm{I}_p,\,\,\,\bm{J}_N\bm{CF}=\bm{CF},\,\,\,
	\where\,\,\,\bm{C}=\bm{U}(\bm{U}^{\prime}\bm{U})^{-1}\bm{U}^{\prime}\nonumber
	\end{gather}
\end{prop}

\begin{proof}
	Using constraints, Equation $(\ref{eq:ap linear const})$ can be rewritten as
	\begin{align}
	\phi^{const}=\frac{1}{N}\tr\bm{F}^{\prime}\bm{CF}-\frac{2}{Nm}\tr\bm{B}^{\prime}\bm{Z}^{\prime}\bm{J}_N\bm{CF}\label{eq:ap const 1}
	\end{align}
	Solving this for $\bm{F}$, we obtain
	\begin{gather*}
	\bm{F}=\frac{1}{m}\bm{J}_N\bm{ZB}
	\end{gather*}
	Inserting this into Equation ($\ref{eq:ap const 1}$), we obtain a minimization problem of
	\begin{align}	
	-\frac{1}{Nm^2}\tr\bm{B}^{\prime}\bm{Z}^{\prime}\bm{J}_N\bm{U}(\bm{U}^{\prime}\bm{U})^{-1}\bm{U}^{\prime}\bm{J}_N\bm{ZB}.\label{eq:min of const}
	\end{align}
	On the other hand, using the proof in Proposition $\ref{prop:groupals and cca}$, the optimization problem in Equation ($\ref{eq:ap grpals}$) can also be rewritten as ($\ref{eq:min of const}$). Thus we obtain the proposition.
\end{proof}

%
\section{Comparison of MSCCA and averaging approach}
\label{sec:how it works}
%
\begin{figure}
	\begin{center}
		\includegraphics[width=13cm]{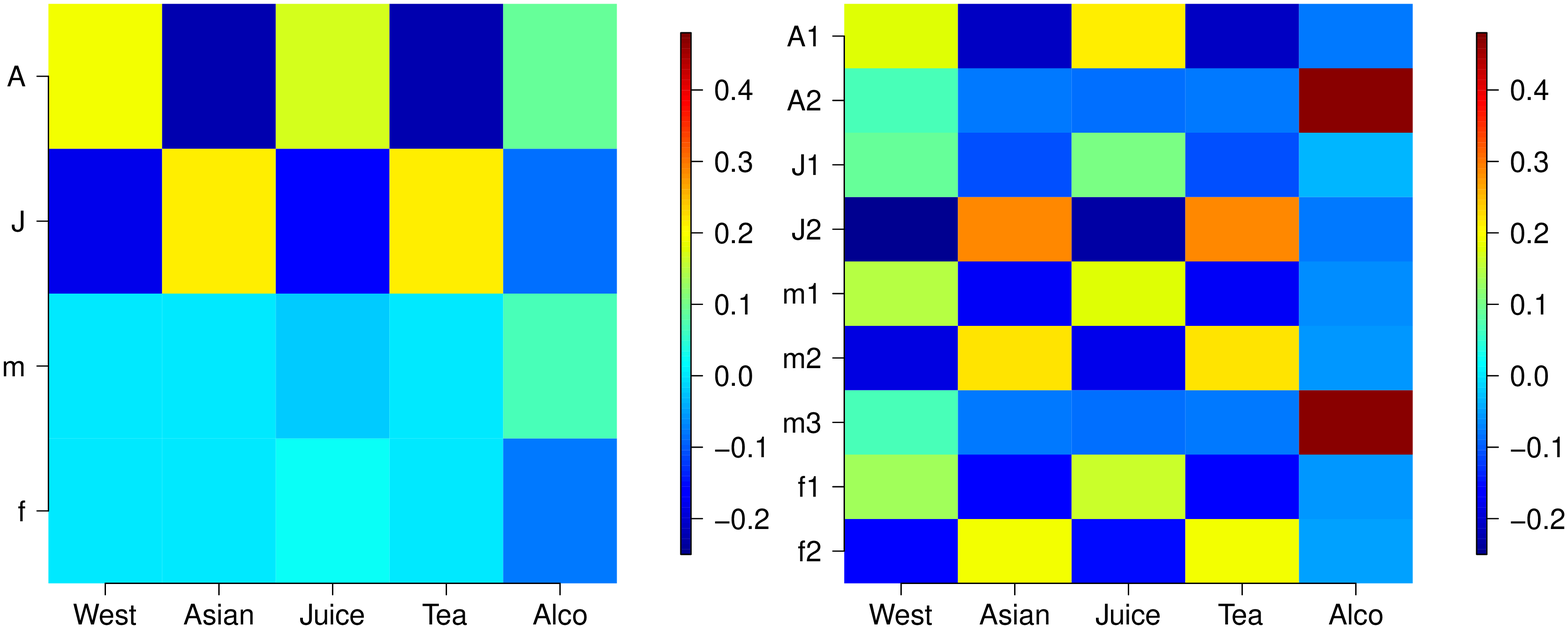}
		\caption{{\small Heatmaps of $\widetilde{\bm{P}}^{ave}$ (left) and $\widetilde{\bm{P}}^{MSCCA}$ (right), respectively. These matrices are calculated as in Equation $(\ref{eq:DPD})$, using $\bm{P}^{ave}=(NHm)^{-1}\bm{V}^{\prime}\bm{Z}^H$ and $\bm{P}^{MSCCA}=(NHm)^{-1}\bm{U}^{\prime}\bm{Z}^H$, respectively.}}
		\label{fig:heatmap}
	\end{center}
\end{figure}

To understand why MSCCA can depict heterogeneous tendencies more clearly than the averaging approach, in this Appendix, we compare the methods that the two approaches use to calculate associations between classes and categories. That is, both MSCCA and averaging reflect a CA framework. Averaging is equivalent to CA for the $\sum_{h=1}^Hr_h \times Q$ contingency table (row is class, column is category); MSCCA is equivalent to CA for the $\sum_{h=1}^H\sum_{s=1}^{r_h}K_{hs} \times Q$ contingency table (row is clusters in each class, column is category), for a given cluster allocation. Figure $\ref{fig:heatmap}$ shows heat maps of the relative deviations, $\widetilde{\bm{P}}^{ave}$ and $\widetilde{\bm{P}}^{MSCCA}$, for each method calculated based on their respective contingency tables. Thus, using this framework, we can say that the difference between the two methods is whether the rows of the contingency table are split by clusters in each class.

This factor then distinguishes between averaging and MSCCA in the calculation of the expected frequency, $\bm{rc}^{\prime}$. Specifically, in the averaging approach, the expected frequency in the (3,1) element in $\widetilde{\bm{P}}^{ave}$ is calculated using the number of individuals who are American and choose ``alcohol", whereas that for the (2,5) element in $\widetilde{\bm{P}}^{MSCCA}$ results from calculating the number of individuals who are in the second cluster in the American class and choose ``alcohol". That is, in MSCCA, the number of individuals used to calculate expected frequency is less for each row in the contingency table than the number for the averaging approach.

Note that the relative deviation indicates the size of the observed frequency (i.e., the number of individuals choosing a particular category), compared with the expected frequency (i.e., the expected number of individuals choosing the category under an assumption of independence). Therefore, the relative deviation tends to increase when the expected frequency is calculated using the limited number of individuals who select the same categories. 

Thus in MSCCA, clustering individuals for each class reveals  the heterogeneous tendencies within each class clearly, regardless of the size of the groups that exhibit similar tendencies.

\end{appendices}

\end{document}